\newif\ifgray

\documentclass[twocolumn,pagebackref]{article}

\usepackage[
	includefoot,
	top=1.6cm,
	bottom=1.3cm,
	left=1.35cm,
	right=1.35cm,
	a4paper,
]{geometry}

\makeatletter
\renewenvironment{abstract}{%
	\small
	\quotation
	\noindent{\bfseries\abstractname.}%
}{\endquotation}
\renewcommand\section{\@startsection {section}{1}{\z@}%
                                   {-2.5ex \@plus -1ex \@minus -.2ex}%
                                   {1.3ex \@plus.2ex}%
                                   {\normalfont\large\bfseries}}
\renewcommand\subsection{\@startsection{subsection}{2}{\z@}%
                                     {-2.25ex\@plus -1ex \@minus -.2ex}%
                                     {0.5ex \@plus .2ex}%
                                     {\normalfont\normalsize\bfseries}}
\renewcommand\subsubsection{\@startsection{subsubsection}{3}{\z@}%
                                     {-2.25ex\@plus -1ex \@minus -.2ex}%
                                     {0.5ex \@plus .2ex}%
                                     {\normalfont\normalsize\bfseries}}
\renewcommand\paragraph{\@startsection{paragraph}{4}{\z@}%
                                    {2.25ex \@plus1ex \@minus.2ex}%
                                    {-1em}%
                                    {\normalfont\normalsize\bfseries}}
\renewcommand\subparagraph{\@startsection{subparagraph}{5}{\parindent}%
                                       {2.25ex \@plus1ex \@minus .2ex}%
                                       {-1em}%
                                      {\normalfont\normalsize\bfseries}}
\makeatother

\usepackage{soul}
\usepackage{url}
\usepackage[x11names, table]{xcolor}
\usepackage[utf8]{inputenc}
\usepackage[small]{caption}
\usepackage{graphicx}
\usepackage{amsmath}
\usepackage{booktabs}
\usepackage{authblk}

\usepackage[T1]{fontenc}
\usepackage{amssymb}
\usepackage[outline]{contour}
\usepackage{enumerate}
\usepackage{dsfont}
\usepackage{tabularx}
\usepackage[ruled,vlined,linesnumbered]{algorithm2e}

\usepackage[square,numbers]{natbib}
\usepackage[linkcolor=red!50!black,citecolor=green!40!black,colorlinks]{hyperref}
\usepackage{doi}
\usepackage{paralist}
\usepackage{tikz}
\usetikzlibrary{calc,positioning}
\usepackage{pgfplots}
\pgfplotsset{compat=1.5}
\pgfplotsset{major grid style={very thin,gray!20!white}}
\usepackage{pgfplotstable}
\usepackage{subcaption}

\usepackage{mathtools}

\usepackage{multirow}

\usepackage[utf8]{inputenc}

\usepackage{enumerate}
\usepackage{amsmath, amssymb, amsthm}
\usepackage{mathtools}
\usepackage{tabularx,booktabs}
\usepackage{dsfont}
\usepackage{xspace}
\usepackage{ mathrsfs }
\usepackage{pifont}
\newcommand{\cmark}{\ding{51}}
\newcommand{\xmark}{\ding{55}}

\usepackage[capitalize,nameinlink,sort&compress]{cleveref}
\crefname{construction}{Construction}{Constructions}
\crefname{claim}{Claim}{Claims}
\crefname{paragraph}{Paragraph}{Paragraphs}
\crefname{observation}{Observation}{Observations}
\crefname{theorem}{Theorem}{Theorems}
\crefname{lemma}{Lemma}{Lemmata}
\crefname{proposition}{Proposition}{Propositions}
\crefname{corollary}{Corollary}{Corollaries}
\crefname{remark}{Remark}{Remarks}
\crefname{section}{Section}{Sections}
\crefname{chapter}{Chapter}{Chapters}
\crefname{figure}{Figure}{Figures}
\crefname{table}{Table}{Tables}
\crefname{definition}{Definition}{Definitions}
\crefname{algorithm}{Algorithm}{Algorithms}
\crefname{equation}{Equation}{Equations}
\crefname{appendix}{Appendix}{Appendices}

\newtheorem{lemma}{Lemma}
\newtheorem{theorem}{Theorem}
\newtheorem{corollary}{Corollary}
\newtheorem{definition}{Definition}
\newtheorem{proposition}{Proposition}

\newtheorem{example}{Example}
\newtheorem{observation}{Observation}{\bf}{\em}

\newcommand{\cs}{\text{Cs}} %
\newcommand{\ds}{\text{Ds}} %
\newcommand{\mds}{\text{mDs}} %
\newcommand{\pop}{\text{d}_\text{pm}} %
\newcommand{\mmc}{\text{MMC}} %
\newcommand{\kt}{\text{d}_\text{Kt}} %
\newcommand{\spdist}{\text{d}_\text{sp}} %

\DeclareMathOperator{\dist}{\kt} %

\usepackage{tikz}
\usetikzlibrary{shapes,arrows, decorations.pathmorphing}

\newcommand{\appsymb}{$\bigstar$}
\newcommand{\appref}[1]{{\appsymb}}

\newcommand{\appendixsection}[1]{%
  \gappto{\appendixProofText}{\section{Additional Material for Section~\ref{#1}}\label{app:#1}}
}

\newcommand{\toappendixx}[1]{}
\newcommand{\toappendix}[1]{%
\gappto{\appendixProofText}
  {{
    #1
  }}
}
\newcommand{\appendixproof}[2]{%
  \gappto{\appendixProofText}
  {
    \subsection{Proof of \cref{#1}}\label{proof:#1}
    #2
  }
}

\usepackage{etoolbox}
\usepackage{marginnote}
\usepackage[textsize=scriptsize]{todonotes} %
\newcommand{\mytodo}[2]{\todo[size=\tiny, color=#1!50!white]{#2}\xspace}
\newcommand{\myrevtodo}[2]{{%
		\let\marginpar\marginnote
		\reversemarginpar
		\renewcommand{\baselinestretch}{0.8}%
		\todo[size=\tiny, color=#1!50!white]{#2}\xspace}}
\newcommand{\myinlinetodo}[2]{\todo[size=\small, color=#1!50!white, inline, caption={}]{#2}\xspace}
\newcommand{\registerAuthor}[3]{%
	\expandafter\newcommand\csname #2com\endcsname[1]{\mytodo{#3}{\textsc{#2}: 
	##1}}%
	\expandafter\newcommand\csname 
	#2revcom\endcsname[1]{\myrevtodo{#3}{\textsc{#2}: ##1}}%
	\expandafter\newcommand\csname 
	#2inline\endcsname[1]{\myinlinetodo{#3}{\textsc{#2}: ##1}}%
	\expandafter\newcommand\csname 
	#2inlineLater\endcsname[1]{\lv{\myinlinetodo{#3}{\textsc{#2}: ##1}}}%
}

\registerAuthor{Robert}{rb}{yellow}
\registerAuthor{Jonas}{ji}{green!50!gray}
\registerAuthor{Leon}{lk}{gray!10!white}
\registerAuthor{Marie}{mg}{green!40!blue!25!white}

\newcommand{\profile}{\ensuremath{\mathcal{P}}}
\newcommand{\rrule}{\ensuremath{\mathcal{R}}}
\newcommand{\peak}[1]{\ensuremath{\top(#1)}}

\title{%
\Large \bf Single-Peaked Opinion Updates
}

\author{Robert Bredereck}
\affil{\normalsize Algorithm Engineering, Humboldt-Universität Berlin, Germany}
\affil{\normalsize Institut für Informatik, TU Clausthal, Germany\\ \texttt{robert.bredereck@tu-clausthal.de}}
\author{Anne-Marie George}
\affil{\normalsize Analytical Solutions and Reasoning, University of Oslo, Germany\\ \texttt{annemage@ifi.uio.no}}
\author{Jonas Israel}
\affil{\normalsize Research Group Efficient Algorithms, Technische Universität Berlin, Germany\\ \texttt{j.israel@tu-berlin.de}}
\author{Leon Kellerhals}
\affil{\normalsize Algorithmics and Computational Complexity, Technische Universität Berlin, Germany\\ \texttt{leon.kellerhals@tu-berlin.de}}

\date{}

\begin{document}

\maketitle

\begin{abstract}
We consider opinion diffusion for undirected networks with sequential updates when the opinions of the agents are single-peaked preference rankings.
Our starting point is the study of  \emph{preserving single-peakedness}.
We identify voting rules that, when given a single-peaked profile, output at least one ranking that is single peaked w.r.t.\ a single-peaked axis of the input.
For such voting rules we show convergence to a stable state of the diffusion process that uses the voting rule as the agents' update rule.
Further, we establish an efficient algorithm that maximises the spread of \emph{extreme} opinions.
\end{abstract}

\section{Introduction}\label{sec: introduction}
	Ahead of elections, but also for competing products on a market, finding a way to maximally spread one particular opinion through exposure of contents to targeted agents has become of increasing interest.
	Advances in understanding the diffusion of opinions help to grasp the extent to which opinions can be manipulated and spread in networks.
	We study opinion diffusion in a setting where agents' opinions are modelled as single-peaked rankings over a set of candidates.
	In each update step one agent observes the preferences of all their neighbours in the network,
	aggregates these by a given voting rule
 	and changes their opinion accordingly. %
 	For issues where preferences are naturally single-peaked, it seems reasonable to assume that also the updated preferences of an agent in a diffusion process remain single-peaked. We investigate which voting rules are applicable in this sense, which lead to converging diffusion dynamics, and whether it is tractable to find update sequences that maximally spread an extreme opinion.
	
	Research has found that computing an optimal sequence of updates to spread a specific opinion is easy for two competing opinions \cite{BE17} but it turns out to be hard in most cases involving multiple independent opinions \cite{AFFG19,AFG20}. However, for other scenarios, notably elections, the agents' opinions are better modeled by rankings over some candidates than by independent opinions.
	In this paper, we explore this additional structure on opinions which allows us to consider various known voting rules (like Kemeny, (weak) Dodgson, and Minimax Condorcet) as update rules.
Since different opinions are not necessarily treated equally under some voting rules, our work significantly differs from work studying opinion diffusion of multiple independent opinions which often use some simple threshold function for the updating process. 
	For example, under some voting rule one opinion (i.e. ranking) might only be adopted if and only if the majority of the agent's neighbours bares this opinion, whereas the same might not be true for other opinions under the same rule. 

	Originally motivated by preference aggregation in context of economic phenomena
	such as prices or quantities~\cite{Black1948}, single-peakedness is probably the most
	prominent restricted preference domain in social choice~\cite{BBHH15,FHH14,FHHR09}. 
	This domain restriction solves many computational and conceptual issues of preference aggregation: among other inviting properties, the aggregation of rankings becomes tractable (e.g. for Kemeny), Condorcet cycles cannot exist, and Arrow's impossibility theorem does not apply anymore.
	While political elections are often not (perfectly) single-peaked,
	preferences in other settings
	often depend on some one-dimensional criterion, such as when
	voting on the temperature in a room,
	choosing the starting time of some event,
	fixing the voting age for an election, or considering an adequate price for a product.
    For Kemeny's rule (and thus the many other rules that
    coincide with Kemeny in the single-peaked domain),  it is easy to see
    that, given single-peaked preferences, also at least one Kemeny outcome ranking is single-peaked.
	We investigate whether the same can be said for other ranking rules. %
	
The paper first provides basic definitions in \Cref{sec: prelim}. 
	The main part of the paper revolves around the following key questions, for each of which we briefly discuss related work here. We conclude in \Cref{sec: conclusion}.

	\paragraph{Which rules preserve single-peakedness? (\cref{sec: preserving sp})}%
	As we consider opinion diffusion of single-peaked preferences, we want to identify ranking rules that allow agents' preferences to remain single-peaked after updates. Under the single-peaked domain, it is known that Condorcet winners exist such that ranking adaptions of Slater's rule~\cite{slater1961inconsistencies}, Ranked Pairs/Tideman~\cite{tideman1987independence}, Beat Path/Schulze~\cite{schulze2011new}, and Split Cycle~\cite{holliday2020split} -- which repeatedly pick Condorcet winners -- coincide with Kemeny's rule. Furthermore, Kemeny's rule preserves single-peakedness \cite{Tru98}.
	We show that the same is true for some rules that do not coincide with Kemeny (in particular Minimax Condorcet and weak Dodgson, as well as Borda's and Copeland's rule when restricted to three candidates). 
	Furthermore, we identify that, in general, Dodgson's rule, Copeland's rule, Borda's rule, and Single Transferable Vote do not preserve single-peakedness.
	\cref{tab:rules} gives an overview over the properties of the listed rules.
\begin{table}[t]
    \centering
    \caption{
    An overview which ranking rules are weak Condorcet winner (CWC) and loser (CLC) consistent, single-peaked preserving (SPP), and extremist majority consistent (EMC).
    See \cref{sec: prelim,sec: general results extr maj conist} for definitions of the properties.
    }
    \begin{tabular}{l c c c c}
        \toprule
        Rule & CWC & CLC & SPP & EMC\\
        \midrule
        Kemeny                  & \cmark & \cmark& \cmark & \cmark \\
        Minimax Condorcet       & \cmark & \xmark& \cmark & \cmark \\
        Weak Dodgson            & \cmark & \xmark& \cmark & \cmark \\
        Dodgson                 & \cmark & \xmark& \xmark & \xmark \\
        Copeland                & \cmark & \cmark& \xmark & \xmark \\
        Single Transferable Vote\!\!\!& \xmark & \xmark& \xmark & \xmark \\
        Borda                   & \xmark & \xmark& \xmark & \xmark \\
        \bottomrule
    \end{tabular}
    \label{tab:rules}
\end{table}
	
	\paragraph{Which rules converge to a stable state? (\cref{sec: convergence})}
	While some works consider convergence of diffusion processes under simultaneous updates of the agents \cite{FKW13,ZWW+20a,CLPT20a}, we focus on sequential updates. 
	For the case of only two opinions a standard update rule is to follow the (strict) majority of the neighbours' opinions.
	Here, the diffusion can be shown to always converge within a bounded number of steps \cite{FKW13}. Considering more than two opinions allows to consider thresholds and in particular majority update rules or even averaging operators in the case of continuous numbers as opinions \cite{Nor20,AFFG19,AFG20}. \citet{BJK20} show that any sequence of majority updates is finite when given $k \in \mathds{N}$ opinions. 
	\citet{FCKT18a} consider rankings as opinions and establish among other results that following majority updates converges. In their model however every vertex corresponds to a cluster of exactly those voters that have the same ranking preference and two such clusters are connected when the rankings are the same up to one swap. 
	Most similarly to our work, \citet{HYT+13} and \citet{BEEG16a} also consider rankings as opinions of individual agents in a network. Both consider random update sequences and allow opinion updates in form of single swaps of adjacent candidates within a voter's ranking. We show that these types of updates can replicate any Kemeny update sequence in the single-peaked domain when replacing one Kemeny update by repeated updates of single swaps by one user. However, \citeauthor{HYT+13} consider only complete networks and \citeauthor{BEEG16a} consider acyclic directed networks and simple directed cycles; hence their convergence results are not transferable to our setting.
	We show that for single-peaked opinions Kemeny and Minimax Condorcet updates converge, i.e., any update sequence is finite.
	
	\paragraph{Is it tractable to find a sequence of updates that promotes the maximal spread of an (extreme) opinion? (\cref{sec: spreading extreme opinions})}
	In line with many works on opinion diffusion we study the maximal spread of one particular opinion.
	While for the case of two opinions a simple greedy approach guarantees a stable state with a maximal spread of one opinion \cite{FKW13}, the problem of finding such a sequence when three or more opinions are present becomes NP-hard for strict majority updates even on very restricted graph structures \cite{BJK20}. In addition, \citet{AFFG19} show that this problem is hard for three opinions when using weak majority updates but identify some tractable graph structures.
	As we consider single-peaked preference rankings as opinions, we focus on the spread of \emph{extreme} opinions, i.e., the opinions that rank the left- or rightmost candidates on the single-peaked axis the highest. Here we show that under voting rules for which an extreme ranking can only be adopted if and only if a (weak) majority of neighbours hold this opinion, a greedy approach 
	similar to the approach for two opinions and majority updates
	can guarantee a maximum number of stable agents with the extreme opinion (but the remaining agents do not necessarily become stable). This case applies to Kemeny and Minimax Condorcet updates as well as updates based on weak Dodgson (as defined in~\cite{Fishburn77}). 
Note that these results only hold since only updates to extreme opinions must follow the weak majority of their neighbours, in contrast to the general hardness results for majority updates on more than two opinions for converging sequences established by  \citet{AFFG19} and \citet{AFG20}.

\section{Preliminaries}\label{sec: prelim}
A {\em preference profile}~$\profile=(C,V)$ consists of a set~$C$ of $m$~candidates
and a set~$V$ of $n$~voters.
Each voter~$v$ has a preference list (or ranking) $\succ_v$ over the candidates.
If, for instance, $C = \{a,b,c\}$, the order~$c \succ_v b \succ_v a$ means that 
voter $v$ prefers $c$ the most and $a$ the least.
We omit the subscript from~$\succ$ if it is clear from context.

\paragraph{Restricted Preferences.}
Central to our work is the concept of single-peakedness, a domain restriction that assumes that voters' preferences are mainly influenced by the position of the candidates on a one-dimensional axis.
\begin{definition}
We call a preference profile $\profile=(C,V)$ \emph{single-peaked} if there is some ordering $\rhd$ of the candidates $C$, the so called \emph{single-peaked axis}, such that for all $c_i, c_j, c_k \in C$ with $c_i \rhd c_j \rhd c_k$ or $c_k \rhd c_j \rhd c_i$ each voter $v \in V$ satisfies $c_j \succ_v c_k$ if $c_i \succ_v c_j$. We call the most preferred candidate of voter $v$'s preference its \emph{peak} $\peak{v}$.
\end{definition}%
\noindent
In the remainder of the paper, when talking about a single-peaked preference profile $\profile = (C,V)$ without further specification, we assume $|C|=m$ and that $\profile$ is single-peaked with respect to the axis $c_1 \rhd \dots \rhd c_m$.
We denote by $\mathcal{L}$ the set of all rankings over candidates $C$ and by $\mathcal{L}_{sp}^\rhd(C)$
the set of single-peaked rankings over $C$ w.r.t.\ the single-peaked axis $\rhd$. The \emph{extreme opinion} $(c_1 \succ \dots \succ c_m) \in \mathcal{L}_{sp}^\rhd(C)$ is denoted by $r^{\uparrow}$ and the opposite extreme $(c_m \succ \dots \succ c_1) \in \mathcal{L}_{sp}^\rhd(C)$ by $r^{\downarrow}$.%

\paragraph{Condorcet Winners and Losers.}
For a preference profile $\profile=(C,V)$, a (weak) Condorcet winner $c$ is a candidate that is not defeated by another candidate in a pairwise comparison, i.e., for all $c' \in C \setminus \{c\}$ we have $|\{v \in V \mid c \succ_v c'\}|\ge|\{v \in V \mid c' \succ_v c\}|$.
A  strict Condorcet winner is a candidate $c$ that strictly defeats all other candidates in a pairwise comparison, i.e., for all $c' \in C \setminus \{c\}$ we have $|\{v \in V \mid c \succ_v c'\}|>|\{v \in V \mid c' \succ_v c\}|$.
The notions of (weak) Condorcet loser and strict Condorcet loser are defined analogously.
While this is not true for general profiles, for single-peaked profiles a weak Condorcet winner and loser always exist and the set of weak Condorcet winners can be determined as indicated by the Median Voter Theorem.

\begin{lemma}[\cite{Black1948}]\label{le: extension-of-median-voter-thm}
Let~$\profile$ be a single-peaked preference profile. Consider the order of the voters' top-ranked candidates along the single-peaked axis. Then the set of median peaks of the voters' preferences and any candidate between median peaks is the set of (weak) Condorcet winners.
\end{lemma}

\paragraph{Ranking Rules.}
A \emph{ranking rule} is an aggregator function $\rrule \colon \mathcal{L}^n \to \mathscr{P}(\mathcal{L})$ (where $\mathscr{P}$ denotes the power set) that aggregates the preference rankings of voters into an output set of preference rankings.
If ranking rule $\rrule$ always outputs rankings that rank a (weak) Condorcet winner highest whenever such a candidate exists,
we call $\rrule$ \emph{(weak) Condorcet winner consistent}.
Similarly, a ranking rule $\rrule$ is called \emph{(weak) Condorcet loser consistent},
if $\rrule$ always outputs rankings that rank a (weak) Condorcet loser lowest whenever one exists.
Note that while it is usually desired to aggregate the preference rankings of  voters into a single output ranking, a property often called \emph{resoluteness}, most ranking rules fail to achieve this.
In this case, one often uses tie-breaking rules to resolve irresoluteness.
Specific ranking rules that we consider more closely and a general tie-breaking rule that is applied throughout this paper are defined at the end of this section.

\paragraph{Diffusion Process.}
For a subset of voters~$V' \subseteq V$ of a preference profile $\profile=(C,V)$ we define 
the preference profile
induced by $V'$ as $\profile[V']=(C,V')$.
A \emph{preference network} is a graph~$G=(V,E)$ with some profile~$(C,V)$
where the voters coincide with the vertices.

Given a preference network~$G=(V,E)$ with preference profile~$\profile=(C,V)$,
we consider the following opinion diffusion process.
At each \emph{update step}, a voter~$v$ (also called the \emph{active} voter) applies a ranking rule $\rrule$ on
the preference profile induced by their neighbourhood $N(v)$ in the
preference network and takes the aggregated ranking as their new opinion.
Formally, we replace the preference ranking of~$v$ 
by a ranking in $\rrule(\profile[N(v)])$.
We speak of update rules instead of ranking rules when considering opinion diffusion processes.

If an active voter $v$ does not change their opinion, i.e., a ranking in $\rrule(\profile[N(v)])$ coincides with $v$'s current preference ranking, then we call the voter $v$ \emph{stable}.
A preference network $G=(V,E)$ with preference profile $\profile=(C,V)$ is in a \emph{stable state} if all the voters in $V$ are stable.
We call a sequence of voters $(v_1, \dots, v_k)$ an \emph{update cycle} in preference network $G$ with preference profile $\profile$ and ranking rule $\rrule$,
when updating voters' opinions along the voter sequence leads to the same preference profile $\profile$ for $G$.
A ranking rule $\rrule$ is said to \emph{converge} if any sequence of voter updates in any preference network $G$ with profile $\profile$ is finite, i.e., results in a stable state. That is, if $\rrule$ converges there cannot exist any update cycles.

\appendixsection{subsec: kemeny sp preserving}
\paragraph{Kemeny's rule.} Kemeny's rule returns those rankings that minimise Kendall's tau distance to all voters' preference rankings.
For a pair of rankings~$\succ,\succ'$, \emph{Kendall's tau distance}
(or ``inversion distance'', since it counts the inversions between two permutations)
between~$\succ$ and~$\succ'$  is defined as
\[\dist(\succ,\succ') = \sum\nolimits_{\{c,c'\} \subseteq C}d_{\succ,\succ'}(c,c'),\]
where $d_{\succ,\succ'}(c,c')$ is set to~$0$ if~$\succ$ and~$\succ'$ rank~$c$
and~$c'$ in the same order, and is set to~$1$, otherwise.  

The \emph{Kemeny score} of a ranking~$r$ with respect to a profile~$\profile = (C,V)$ is defined as~$\dist(r, \profile):=\sum_{v \in V} \dist(r,\succ_v)$.
A ranking~$r$ with a minimal Kemeny score is called a \emph{Kemeny ranking} of~$\profile$
and its score is the \emph{Kemeny score} of the profile.

\citet{BEEG16a} propose an update procedure by which voters swap the positions of two adjacent candidates in their ranking whenever the (strict) majority of their neighbours orders the two candidates in this way.
When all voters' preferences are single-peaked, their update rule can replicate Kemeny's rule by making every voter execute such updates repeatedly until they reach a fixed preference. 
However, as the work of Brill et al.\ concerns opinion diffusion in directed graphs (that are acyclic or simple cycles), their results are not transferable to our setting.
Further details on this relation are deferred to \cref{apsec: Kemeny-and-Brilletal-update-relation}.

\toappendix
{
\subsection{Relation of Kemeny Updates and Brill et al.'s Update Rule}
\label{apsec: Kemeny-and-Brilletal-update-relation}
\citet{BEEG16a} propose an update procedure by which voters swap the positions of two adjacent candidates in their ranking whenever the (strict) majority of their neighbours orders the two candidates in this way.
As this update rule bares great similarity to Kemeny's rule, we analyse whether it can replicate Kemeny updates.

The following example shows, that this is not the case in general.
\begin{example}
Consider a voter in a preference network with opinion $a \succ c \succ b$ that has four neighbours with the following opinions.
    \begin{align*}
        1 \times ~~ &a \succ b \succ c \\
        2 \times ~~ &b \succ a \succ c \\
        1 \times ~~ &a \succ c \succ b 
    \end{align*}
    
The Kemeny ranking in this case is $b \succ a \succ c$. However, the only local update that the voter can make is to swap candidates $c$ and $b$ to the opinion $a \succ b \succ c$, since three out of the four neighbours agree on $b \succ c$. No further updates are then possible.
\end{example}

However, as we argue in the following, in a single-peaked domain Brill et al.'s update rule can replicate updates based on Kemeny's rule, when a voter updates their opinion by pairwise swaps repeatedly until no further update is possible. 

Kemeny's rule is weak Condorcet loser consistent and a weak Condorcet loser always exists in a single-peaked domain. In the case of two candidates this is obviously true. For an inductive argument, assume, that the statement is true for $m$ candidates. Let $\profile$ be a single-peaked preference profile of the voter's neighbours over $m+1$ candidates. If a single Condorcet loser $c$ exists in $\profile$, then swaps can be performed repeatedly to place this Condorcet loser at the bottom of the voters preference ranking, because $c$ is beaten by any other candidate in a majority comparison. Because Kemeny's rule is Condorcet loser consistent, all Kemeny rankings of $\profile$ consist of Kemeny rankings of $\profile \setminus \{c\}$ that are extended by $c$ at the bottom. But since by our induction hypothesis any Kemeny update of a voter over $m$ candidates can be replicated by pairwise swaps following the strict majority opinion of the neighbours, the same holds for the Kemeny update of the voter with neighbourhood profile $\profile$.
Now assume that both extreme candidates $c_1, c_m$ are weak Condorcet losers in $\profile$. Then because of tie-breaking based on Kendall's tau distance, the Kemeny update for the voter will rank $c \in \{c_1, c_m\}$ last if and only if $c$ is ranked last in the voters own current opinion. Because of the induction hypothesis and because Kemeny's rule is weak Condorcet loser consistent, the voter can perform pairwise swaps following the majority opinion to gain the Kemeny update ranking.

As mentioned in the main paper, while Brill et al.'s~\cite{BEEG16a} update rule can replicate Kemeny updates, they consider this rule only on directed networks.
}

\paragraph{Minimax Condorcet.}
\appendixsection{subsec: MMC sp preserving}

Minimax Condorcet (MMC), also known as Simpson's rule, ranks candidates by their worst defeat in pair-wise comparisons. Let $\pop(c',c) \coloneqq |\{v \in V \mid c' \succ_v c \}| - |\{v \in V \mid c \succ_v c' \}|$ be the \emph{popularity margin} between $c, c' \in C$. Then MMC ranks candidates $c \in C$ in non-decreasing order by their \emph{MMC score} $\mmc(c) \coloneqq \max_{c' \in C} \pop(c',c)$.
MMC differs from Kemeny's rule even in the single-peaked domain, e.g., by violating Condorcet loser consistency.
\begin{observation}[\appref{obs:minmaxcond-cl-inconsistent}$^1$]
	\label{obs:minmaxcond-cl-inconsistent}
	MMC is not Condorcet loser consistent even in the single-peaked domain.
\end{observation}
\footnotetext[1]{Proofs of results marked with \appref{} are deferred to the appendix.}
\appendixproof{obs:minmaxcond-cl-inconsistent}{
\begin{proof}
    Consider the single-peaked preference profile with axis~$a\rhd b\rhd c\rhd d$, given by
    \begin{align*}
        1 \times ~~ &a \succ b \succ c \succ d \\
        2 \times ~~ &b \succ c \succ d \succ a
    \end{align*}
    The resulting scores are $\mmc(a) = 1$, $\mmc(b)=-1$, $\mmc(c) = 3$ and $\mmc(d)=3$. Thus either candidate $c$ or candidate $d$ gets ranked last by MMC even though $a$ is the strict Condorcet loser.
\end{proof}
}

\paragraph{Tie-Breaking.}
The ranking rules we consider are in general irresolute. Since the active voter in the opinion diffusion process can change only to a single new opinion we need a way to break ties. First, we always restrict the set of winning rankings to those that are single-peaked w.r.t.\ the given single-peak axis, if existent. 
Further, we envision the active voter to tend to that single-peaked ranking for which they do not need to adapt their current opinion too much, i.e., we break possible further ties in favor of the single-peaked ranking that minimises Kendall's tau distance to the current opinion of the active voter.
If still two or more rankings are tied by their Kendall's tau distance, we break ties arbitrarily, e.g., in lexicographical order with respect to the single-peaked axis.

\section{Preserving Single-Peakedness}\label{sec: preserving sp}
In this section, we consider the question of whether a ranking rule preserves single-peakedness, i.e., outputs at least one single-peaked ranking when given a single peaked profile. This is not only relevant for our later analysis of diffusion of single-peaked opinions, but also extends the existing research on this domain restriction.

\begin{definition}
    A ranking rule $\rrule$ preserves single-peakedness if, given a single-peaked preference profile $\profile=(C,V) \in \mathcal{L}_{sp}^\rhd(C)$ with axis $\rhd$, it always outputs at least one single-peaked ranking w.r.t. $\rhd$, i.e., $\rrule(\profile) \cap \mathcal{L}_{sp}^\rhd(C) \neq \emptyset$.
\end{definition}

We next investigate which ranking rules preserve single-peakedness.
Intriguingly, our results show that preserving single-peakedness is independent of being (weak) Condorcet winner or loser consistent:
Copeland's rule is weak Condorcet winner and loser consistent but does not preserve single-peakedness.
Conversely, Minimax Condorcet preserves single-peakedness but is not Condorcet loser consistent.
It is easy to construct rules that preserve single-peakedness but are not Condorcet winner consistent.
One such example is the rule that always returns the ranking~$r^\downarrow$.

\subsection{Kemeny's Rule and its Equivalents on the Single-Peaked Domain}\label{subsec: kemeny sp preserving}

Kemeny's rule preserves single-peakedness~\cite{CGS13,Tru98,BetzlerBN14}. For self containment we provide an induction-based proof in the appendix.
\begin{observation}[\appref{obs:kemeny-sp}]
	\label{obs:kemeny-sp}
    Kemeny's rule preserves single-peaked\-ness.
\end{observation}
\appendixproof{obs:kemeny-sp}{
\begin{proof}
    Our proof is by induction over the number of candidates $m$.
    As any ranking of~$m=2$ candidates is single-peaked, this is our base case.
    Assume that the statement is true for $m$ candidates
    and let $\profile$ be a single-peaked preference profile over the axis $c_1 \rhd \dots \rhd c_m \rhd c_{m+1}$.
    For any single-peaked order over the $m+1$ candidates, either $c_1$ or $c_{m+1}$ must be ranked last.
    Hence the set of (weak) Condorcet losers $CL$ contains at least one of $c_1$ or $c_{m+1}$.
    For a candidate~$c \in CL$ let~$\profile_c$ be the profile obtained by removing~$c$ from all rankings in~$\profile$.
    We can construct the set of all Kemeny rankings of $\profile$ by extending the Kemeny rankings of $\profile_c$ with $c$ as a last ranked candidate for all $c \in CL$.
    When deleting any candidate, in particular any (weak) Condorcet loser $c \in CL$, from the preference profile $\profile$, single-peakedness of all preference rankings is still preserved.
    Let $c \in CL \cap \{c_1,c_m\}$.
    By our assumption, there exists a single-peaked Kemeny ranking for $\profile_c$, $c \in CL$.
    Extending this ranking with $c$ as a last ranked candidate preserves its single-peakedness. Thus, there exists a single-peaked Kemeny ranking of $\profile$.
   \end{proof}
}

Many social choice functions can be defined by specifying their selection process on (weighted) majority graphs.
By the Median Voter Theorem~(\cref{le: extension-of-median-voter-thm}), a majority graph under single-peaked preferences is acyclic.
Thus any two social choice functions that can be defined over (weighted) majority graphs and only differ in the way they handle cycles in the graph will output the same winning set under single-peaked preferences. 
Furthermore, we can adapt any social choice function into a ranking rule by repeatedly deleting one of the (possibly multiple) winners from the profile and appending it to the ranking. %
It is easy to see that, in the single-peaked domain, the ranking adaptions of Slater's rule~\cite{slater1961inconsistencies}, Ranked Pairs/Tideman~\cite{tideman1987independence}, Beat Path/Schulze~\cite{schulze2011new} and Split Cycle~\cite{holliday2020split} coincide with Kemeny's rule.

\subsection{Minimax Condorcet}\label{subsec: MMC sp preserving}

To prove that MMC preserves single-peakedness we use the following two lemmas.
The first shows that when computing the MMC score of a candidate we only need to consider the candidate's direct neighbours on the single-peaked axis. The second provides us with a structure on the set of all rankings that are single-peaked with respect to the same axis. 
\begin{lemma}[\appref{lem:mmc_worst_defeat}]\label{lem:mmc_worst_defeat}
Let $\profile$ be a single-peaked profile on the axis $c_1 \rhd \dots \rhd c_m$. Then any $c_i$, $i \in [m]$, experiences a worst defeat in $\profile$ against $c_{i-1}$ or $c_{i+1}$ (if they exist),
i.e., $\pop(c_\ell, c_i) \ge \pop(c_j, c_i)$ 
for all~$j \in [m]\setminus \{i\}$ and some~$\ell \in \{i-1, i+1\}$.
\end{lemma}
\appendixproof{lem:mmc_worst_defeat}{
\begin{proof}
Towards a contradiction assume that in $\profile$ a candidate $c_i$ experiences their worst defeat against some $c_{i+k}$ for $k \geq 2$ but not against $c_{i+1}$ (the case for $c_{i-k}$ is analogous).
Then there has to be at least one voter $v \in V$ such that $c_{i+k} \succ_v c_i$ and $c_i \succ_v c_{i+1}$. This violates single-peakedness of~$\profile$.
\end{proof}
} %

\begin{lemma}[\appref{lem:mmc_ranking_sequence}]\label{lem:mmc_ranking_sequence}
For an axis $c_1 \rhd \dots \rhd c_m$ there exists an ordering of all single-peaked rankings $r_1, \ldots, r_{2^{m-1}} \in \mathcal{L}_{sp}^\rhd(C)$ such that for all $c_{i}, c_{i+1} \in C$, $i \in [m-1]$, there exists a threshold $h \in [2^{m-1}]$ with
\begin{align*}
    c_i &\succ_r c_{i+1} ~~~\text{for all rankings} ~ r \in \{r_1, \ldots, r_h\} ~~~ \text{ and} \\
    c_{i+1} &\succ_r c_{i} ~~~~~~~\text{for all rankings} ~ r \in \{r_{h+1}, \ldots, r_{2^{m-1}}\}.
\end{align*} 
Further, $r_1, \ldots, r_h$ are exactly those single-peaked rankings where the candidates $c_1, \ldots, c_i$ are the peaks.
\end{lemma}

\begin{proof}[Proof sketch]
We construct the ordering iteratively.
For $m=3$ choose
$(c_1 \succ c_2 \succ c_3, ~~ c_2 \succ c_1 \succ c_3, ~~ c_2 \succ c_3 \succ c_1, ~~ c_3 \succ c_2 \succ c_1).$
We then iteratively include candidates at the end of the single-peaked axis. For this we take the rankings of the previous ordering (over $m-1$ candidates) and replace them with an ordering of rankings where the new candidate $c_{m}$ is included in every possible position (not violating single-peakedness, starting with the highest position possible. Lastly we append the sequence with $c_{m} \succ c_{m-1} \succ \dots \succ c_1$.
\end{proof}

\appendixproof{lem:mmc_ranking_sequence}{
\begin{proof}{}
We prove the claim constructively by providing an algorithm to construct the desired ordering iteratively.
For $m=1$ and $m=2$ the task is trivial. For $m=3$ candidates a possible ordering of all single-peaked rankings satisfying the above properties is given by $(c_1 \succ c_2 \succ c_3, ~~ c_2 \succ c_1 \succ c_3, ~~ c_2 \succ c_3 \succ c_1, ~~ c_3 \succ c_2 \succ c_1)$.
We now describe how to obtain such an ordering of rankings over $m$ candidates given a sequence of rankings over $m-1$ candidates $c_1 \succ \dots \succ c_{m-1}$. W.l.o.g.\ we can add the new candidate $c_{m}$ to the very end of the single-peaked axis such that the new axis is $c_1 \succ \dots \succ c_{m-1} \succ c_{m}$. Going from the first ranking in the ordering to the last we do the following. For each such ranking of $m-1$ candidates we generate rankings for $m$ candidates by adding $c_{m}$ into every possible position lower than the position of $c_{m-1}$. Then we append these rankings to the sequence of rankings in order of decreasing position of $c_{m}$. Finally, we append one additional ranking to the end of the sequence, namely $c_{m} \succ c_{m-1} \succ \ldots \succ c_1$.

It is easy to see that the rankings generated this way are single-peaked w.r.t.\ the given axis. 
Furthermore, since the relative order of all other candidates $c_1, \ldots, c_{m-1}$ is not changed, the claim still holds. To see this fix some $i \in [m-2]$ and consider the two candidates $c_i$ and $c_{i+1}$. Let $r$ be a ranking over the first $m-1$ candidates for which $c_i \succ_r c_{i+1}$ and let $r'$ be a ranking over $m$ candidates constructed by the algorithm from $r$. Then we know that still $c_i \succ_{r'} c_{i+1}$ since the relative order of $c_i$ and $c_{i+1}$ was not changed. The same is true for all rankings on $m-1$ candidates where $c_{i+1} \succ_r c_i$. Thus the newly constructed sequence of rankings over $m$ candidates will again have a threshold in the above sense.
It remains to show that the same is true for $i=m-1$. But since except for the last ranking the candidate $c_{m}$ always got added below $c_{m-1}$ in each ranking, we have $c_{m-1} \succ_r c_{m}$ for all newly constructed rankings except for the last one. Thus the threshold for $c_{m-1}$ and $c_{m}$ is $2^{m-1}-1$.
The claim about the peaks of the rankings $r_1, \dots, r_h$ follows by construction.
\end{proof}
}

We can now prove that MMC always outputs at least one single-peaked ranking if the input profile is single-peaked. Assuming that this is not the case, we reach a contradiction by analysing the MMC scores of three neighbouring candidates that are creating a dip with the help of the order of single-peaked rankings constructed in \cref{lem:mmc_ranking_sequence}. 

\begin{theorem}[\appref{thm:mmc_sp_preserving}]\label{thm:mmc_sp_preserving}
MMC preserves single-peakedness.
\end{theorem}
\appendixproof{thm:mmc_sp_preserving}{
\begin{proof}
Let $c_1 \rhd \dots \rhd c_m$ be a single-peak axis and $\profile$ a profile that is single-peaked on this axis. Further, let $(r_j)_{j \in [2^{m-1}]}$ be the sequence of all single-peaked rankings w.r.t.\ this axis given by \Cref{lem:mmc_ranking_sequence}.
Towards a contradiction assume that under $\profile$ MMC only outputs rankings violating single-peakedness. Let $r \in \rrule_{\mmc}(\profile)$ be such a ranking violating single-peakedness. It therefore exists an $i \in [m]$ such that 
\begin{align}
    \mmc(c_{i-1}) &< \mmc(c_i) ~~~ \text{ and } \label{eq:mmc_SP_preserving_1}\\
    \mmc(c_{i+1}) &< \mmc(c_i). \label{eq:mmc_SP_preserving_2}
\end{align}
Let $h_{-1}$ be the threshold as described in \Cref{lem:mmc_ranking_sequence} for candidates $c_{i-1}$ and $c_i$, and $h_{+1}$ the one for $c_i$ and $c_{i+1}$. Similarly, we define $h_{-2}$ for $c_{i-2}$ and $c_{i-1}$, and $h_{+2}$ for $c_{i+1}$ and $c_{i+2}$. These thresholds satisfy $h_{-2} \leq h_{-1} \leq h_{+1} \leq h_{+2}$, since otherwise there would be at least one ranking in $(r_j)_{j \in [2^{m-1}]}$ that is in violation of the single-peaked axis.

Now consider the above two inequalities. \Cref{eq:mmc_SP_preserving_1} yields \[\max_{b \in C} \pop(b, c_{i-1}) < \max_{b \in C} \pop(b, c_i)\] which by \Cref{lem:mmc_worst_defeat} is equivalent to 
\begin{align*}
    \max &\left\{ \pop(c_{i-2}, c_{i-1}), \pop(c_{i}, c_{i-1}) \right\} \\
    &< \max \left\{ \pop(c_{i-1}, c_i), \pop(c_{i+1}, c_i) \right\}.
\end{align*}
This inequality is true if and only if any of the following two pairs of inequalities is true:
\begin{align}
    \begin{split}
        \pop(c_{i-2}, c_{i-1}) &< \pop(c_{i-1}, c_i) \text{ ~~~and~~~ } \label{eq:mmc_sp_preserving_3} \\
        \pop(c_{i}, c_{i-1}) &< \pop(c_{i-1}, c_i) 
    \end{split} \\
    \begin{split}
        \text{or~~~ } \pop(c_{i-2}, c_{i-1}) &< \pop(c_{i+1}, c_i) \text{ ~~~and~~~ } \\
        \pop(c_{i}, c_{i-1}) &< \pop(c_{i+1}, c_i).\label{eq:mmc_sp_preserving_4}
    \end{split}
\end{align}

For every $r_k \in (r_j)_{j \in [2^{m-1}]}$ let $\rho_k \geq 0$ be the number of occurrences of ranking $r_k$ in $\profile$. By using the definition of the thresholds in \Cref{lem:mmc_ranking_sequence} we know that for example
\begin{align*}
    \pop(c_{i-2}, c_{i-1}) &< \pop(c_{i-1}, c_i) \\ 
    \Leftrightarrow ~~~ \sum_{k=1}^{h_{-2}} \rho_k - \sum_{k=h_{-2}+1}^{2^{m-1}} \rho_k &< \sum_{k=1}^{h_{-1}} \rho_k - \sum_{k=h_{-1}+1}^{2^{m-1}} \rho_k.
\end{align*}
Using this as well as $h_{-2} \leq h_{-1} \leq h_{+1}$ we manipulate \Cref{eq:mmc_sp_preserving_3,eq:mmc_sp_preserving_4} as follows. First consider the latter inequality of  \Cref{eq:mmc_sp_preserving_4}. This is equivalent to
\[ \sum_{k = h_{-1}+1}^{2^{m-1}} \rho_k - \sum_{k = 1}^{h_{-1}} \rho_k < \sum_{k = h_{+1}+1}^{2^{m-1}} \rho_k - \sum_{k = 1}^{h_{+1}} \rho_k. \]
This can be reordered to
\[ \sum_{k = h_{-1}+1}^{h_{+1}} \rho_k < - \sum_{k = h_{-1}+1}^{h_{+1}} \rho_k, \]
which is a contradiction.
Thus \Cref{eq:mmc_sp_preserving_4} is not true and both inequalities in \Cref{eq:mmc_sp_preserving_3} have to be true. In particular, from the latter inequality of \Cref{eq:mmc_sp_preserving_3} we can obtain that
\[ \sum_{k = h_{-1}+1}^{2^{m-1}} \rho_k - \sum_{k = 1}^{h_{-1}} \rho_k < \sum_{k = 1}^{h_{-1}} \rho_k - \sum_{k = h_{-1}+1}^{2^{m-1}} \rho_k. \]
Reordering this inequality we conclude that 
\begin{equation}\label{eq:mmc_sp_preserving_case1}
    \sum_{k = h_{-1}+1}^{2^{m-1}} \rho_k < \sum_{k=1}^{h_{-1}} \rho_k.
\end{equation}
Similarly, \Cref{eq:mmc_SP_preserving_2} yields $\max_{b \in C} \pop(b, c_{i+1}) < \max_{b \in C} \pop(b, c_i)$ which by \Cref{lem:mmc_worst_defeat} is equivalent to 
\begin{align*}
    \max &\left\{ \pop(c_{i+2}, c_{i+1}), \pop(c_{i}, c_{i+1}) \right\} \\
    &< \max \left\{ \pop(c_{i-1}, c_i), \pop(c_{i+1}, c_i) \right\}.
\end{align*}
We again evaluate the different cases similarly as we did with \Cref{eq:mmc_SP_preserving_1} and by $h_{-1} \leq h_{+1} \leq h_{+2}$ obtain 
\[ \sum_{k=1}^{h_{+1}} \rho_k < \sum_{k = h_{+1}+1}^{2^{m-1}} \rho_k, \]
which is equivalent to
\[ \sum_{k=1}^{h_{-1}} \rho_k + \sum_{k=h_{-1}+1}^{h_{+1}} \rho_k < \sum_{k = h_{+1}+1}^{2^{m-1}} \rho_k. \]
If we plug in \Cref{eq:mmc_sp_preserving_case1} here we obtain
\begin{align*}
    \sum_{k=h_{-1}+1}^{2^{m-1}} \rho_k + \sum_{k=h_{-1}+1}^{h_{+1}} \rho_k &< \sum_{k=1}^{h_{-1}} \rho_k + \sum_{k=h_{-1}+1}^{h_{+1}} \rho_k \\
    &< \sum_{k = h_{+1}+1}^{2^{m-1}} \rho_k.
\end{align*}
But this leads to a contradiction as by $h_{-1} \leq h_{+1}$ and $\rho_k \geq 0$ for all $k$ we know that $\sum_{k=h_{-1}+1}^{2^{m-1}} \rho_k \geq \sum_{k = h_{+1}+1}^{2^{m-1}} \rho_k$.
\end{proof}
 }

\subsection{Other Rules}\label{sec: other results}
\appendixsection{sec: other results}
We show that many other ranking rules that do not coincide with the above-mentioned rules fail to preserve single-peakedness.
We refer the reader to the appendix for examples, proofs, and more details.
To summarise, we identify counterexamples for Dodgson's and Copeland's rule (with 5 candidates), Single Transferable Vote (with 3 candidates), and Borda's rule (with 4 candidates), which can be easily extended to examples with more candidates.

\begin{proposition}
	\sloppy
    Dodgson's, Copeland's and Borda's rule, and Single Transferable Vote do not preserve single-peakedness.
\end{proposition}

However, we can show that on 3-candidate profiles Borda's rule preserves single-peakedness.
Copeland's rule with 3 candidates coincides with Kemeny's rule due to weak Condorcet loser and winner consistency.
It thus preserves single-peakedness on these profiles.
Note that any ranking on 2 candidates is single-peaked;
thus under the here discussed rules the only open cases are for Copeland's rule on profiles with 4 candidates and for Dodgson's rule with 3 or 4 candidates. 

We remark that weak Dodgson~\cite{Fishburn77}, in which one ranks candidates by the number of pairwise swaps needed to make a candidate a \emph{weak} Condorcet winner, does preserve single-peakedness.
We analyse the rule in \Cref{subsec: modification of dodgson}.

\toappendix{
\subsection{Dodgson}\label{subsec: dodgson}
Given a preference profile $\profile$, for each candidate $c$ we can compute its Dodgson score $\ds(c)$ w.r.t.\ $\profile$ as the minimum number of pairwise swaps in rankings of $\profile$ to make $c$ the strict Condorcet winner. If a candidate is already the strict Condorcet winner in $\profile$ its Dodgson score is 0.
Dodgson's rule takes as an input a preference profile $\profile = (C,V)$ and outputs a ranking of all $c \in C$ in non-increasing order w.r.t.\ $\ds(c)$.

The following example shows that Dodgson's rule does not preserve single-peakedness.
\begin{example}\label{ex: Dodgson not sp preserving}
    Consider the preference profile given by
    \begin{align*}
        1 \times ~~ &a \succ b \succ c \succ d \succ e \\
        2 \times ~~ &b \succ a \succ c \succ d \succ e \\
        2 \times ~~ &d \succ e \succ c \succ b \succ a \\
        1 \times ~~ &e \succ d \succ c \succ b \succ a.
    \end{align*}
    The resulting Dodgson scores are $\ds(a) = 6$, $\ds(b)=3$, $\ds(c) = 4$, $\ds(d) = 3$ and $\ds(e)=6$. 
    Thus, single-peakedness is violated by all existing Dodgson rankings in this instance since they all have a dip at $c$.
\end{example}
Since this example can be arbitrarily extended by adding more candidates always in the same order to the bottom of each ranking, it generally holds that Dodgson's rule is not preserving single-peakedness on preference profiles with $m\geq5$ candidates.
}

\toappendix{
\subsection{Weak Variant of Dodgson's Rule}\label{subsec: modification of dodgson}

Instead of considering the Dodgson score $\ds(c)$ w.r.t. a preference profile $\profile$, i.e. the minimum number of pairwise swaps in rankings of $\profile$ to make $c$ the \emph{strict} Condorcet winner, we consider a weak version $\mds(c)$ which is the minimum number of pairwise swaps in rankings of $\profile$ to make $c$ a \emph{weak} Condorcet winner. If a candidate is already a weak Condorcet winner in $\profile$ its weak Dodgson score is 0.
The weak Dodgson rule takes as an input a preference profile $\profile$ and outputs all rankings of candidates $c$ in non-increasing order w.r.t.\ $\mds(c)$. 

The following example shows that even on single-peaked preferences the weak Dodgson rule (as well as the original Dodgson rule) is not Condorcet loser consistent. This makes both rules different from Kemeny's rule.
\begin{example}\label{ex: Dodgson-not-CL-consistent}
    Consider the preference profile given by
    \begin{align*}
        1 \times ~~ &a \succ b \succ c \succ d \\
        2 \times ~~ &b \succ c \succ d \succ a.
    \end{align*}
    The resulting scores are $\ds(a) = \mds(a) = 3$, $\ds(b) = \mds(b) = 0$, $\ds(c) = \mds(c) = 2$ and $\ds(d) = \mds(d) = 4$. Thus candidate $d$ gets ranked last by the (weak) Dodgson rule even though $a$ is the strict Condorcet loser.
\end{example}

It is possible to show that the weak Dodgson rule preserves single-peakedness.
Using this we can also show that weak Dodgson is extremist majority consistent. 

\begin{proposition}\label{prop: Dodgson-preserves-sp}
Weak Dodgson perserves single-peakedness.
\end{proposition}
\begin{proof}
Let $c_1 \rhd \cdots \rhd c_m$ be a single-peak axis and $\profile$ a profile that is single-peaked on this axis.
Let~$\succeq_D$ be the order relation on the candidates based on their weak Dodgson scores.
The statement is trivial for $m = 1,2$, since any preference order is single-peaked over one or two candidates. Hence, we consider a profile with $m \geq3$ from now on.
Suppose $\succeq_D$ has a dip with lowest point $c_i$ such that there exists no tie breaking rule that can resolve this dip, i.e., $c_{i-1} \succ_D c_i$ and $c_{i+1} \succ_D c_i$. 

Assume a majority of the voters supports $c_{i-1} > c_i$, and let us call these voters $M \subseteq N$. Because of single-peakedness, the voters $M$ must then also support $c_i > c_l$ for $l = i+1, \dots, m$ and in particular $c_i > c_{i+1}$.
Because $c_{i+1} \succ_D c_i$, $\mds(c_i)\neq 0$ and there must exist at least one candidate that beats $c_i$ by a strict majority in pairwise comparisons. We call the set of such candidates $C'$.
Since $c_i$ beats all $c_l$ with $l = i+1, \dots, m$ by a majority, $C' \subseteq \{c_1, \dots, c_{i-1}\}$. 
But then the set of voters (a strict majority) that prefer any $c' \in C'$ over $c_i$ must also prefer $c' > c_i > c_{i+1}$ because of single-peakedness.
These $C'$ are the only candidates that $c_i$ must beat to become a weak Condorcet winner, but also $c_{i+1}$ has to beat these to become a weak Condorcet winner and needs strictly more pairwise swaps than $c_i$ to do so. Thus $\mds(c_i) < \mds(c_{i+1})$ which is a contradiction to $c_{i+1} \succ_D c_i$.
Thus, a strict majority of voters must support $c_i > c_{i-1}$. Furthermore, because of single-peakedness, the same strict majority of voters must also prefer $c_i > c_h$ for all $h = 1, \dots, i-1$.

By symmetry the same is true for $c_i$ and $c_{i+1}$, i.e., a strict majority of voters must support $c_i > c_{i+1}$. Again because of single-peakedness, the same strict majority of voters must also prefer $c_i > c_h$ for all $h = i+1, \dots, m$.

This means that $c_i$ strictly dominates all other candidates in pairwise majority comparisons, i.e., $c_i$ is a strict Condorcet winner and $\mds(c_i)=0$, which is a contradiction.
   \end{proof}
\begin{proposition}
The weak Dodgson rule is extremist majority consistent.
\end{proposition}
\begin{proof}
We only consider the case for $r^\uparrow$ here, as the case for $r^\downarrow$ follows by symmetry.
Suppose a weak majority  of voters in $\profile$ has opinion $r^{\uparrow}$, then $c_1$ has a weak Dodgson score of $0$. However, $c_1$ might only be a weak Condorcet winner and some other candidates $C'$ also have score $0$. Let $r$ be a single-peaked weak Dodgson ranking (which exists by Proposition~\ref{prop: Dodgson-preserves-sp}). Then the top candidates in the ranking must be candidates $C'$ and suppose candidate $c_i$ is top ranked. Then also $c_1, \dots, c_{i-1} \in C'$ and $c_1 \succ_r \dots \succ_r c_{i-1}$. We can now modify the ranking $r$ by reversing the order over $c_1, \dots, c_{i}$ to gain order $r'$.  $r'$ is a Dodgson ranking, because $c_1, \dots, c_{i}$ all have a Dodgson score of $0$.
Furthermore, because $r$ is single-peaked, also $r'$ must be single-peaked.
However, the only single-peaked order with $c_1$ at the top is $r^{\uparrow}$ and thus $r^{\uparrow}$ is a single-peaked weak Dodgson ranking.

Suppose $r^{\uparrow}$ is a weak Dodgson ranking. Then $c_1$ is one of the candidates with minimal weak Dodgson score. Since $\profile$ is single-peaked there exists a weak Condorcet winner. This candidate has to have minimal weak Dodgson score. Since $c_1$ therefore also must have score 0 it must be a weak Condorcet winner in $\profile$. This means $c_1$ is at least as good as all other candidates in pairwise comparison. The only preference ranking in which $c_1$ beats $c_2$ is $r^{\uparrow}$. Thus there must exist at least a weak majority of voters in $\profile$ with preference $r^{\uparrow}$.
   \end{proof}
}

\toappendix{
\subsection{Copeland}\label{subsec: copeland}

 For a given preference profile $\profile$, the Copeland score of a candidate $c$ is defined as the number of candidates that $c$ beats in a pairwise majority comparison plus half the number of candidates with whom $c$ has a tie, i.e., $\cs(c) = |\{x \in C \mid c \succ_M x\}| + \frac{1}{2}|\{x \in C \mid c \equiv_M x\}|$ where $\succeq_M$ is the majority preference relation of voters in $\profile$. Copeland's rule ranks candidates in non-increasing order of their Copeland scores $\cs$.

\Cref{ex: Dodgson not sp preserving}, by which we showed that Dodgson's rule is not preserving single-peakedness, can also be used to show that Copeland's rule does not preserve single-peakedness. Here, the Copeland scores are $\cs(a) = 1.5$, $\cs(b)=2.5$, $\cs(c) = 2$, $\cs(d) = 2.5$ and $\cs(e)=1.5$. 
    Thus, single-peakedness is violated by all existing Copeland rankings, since they have a dip at $c$. 
This example can be arbitrarily extended by adding more candidates always in the same order to the bottom of each ranking. Thus, it generally holds that Copeland's rule is not preserving single-peakedness on preference profiles with $m\geq5$ candidates.

Since Copeland's rule is weak Condorcet winner and weak Condorcet loser consistent (which always exist in a single-peaked profile), it coinciades with Kemeny's rule on preference profiles over $\leq 3$ candidates and the results of \Cref{subsec: kemeny sp preserving} apply. }

\toappendix{
\subsection{STV}\label{subsec: STV}
STV (Single Transferable Vote) is known as a single winner rule which iteratively eliminates the plurality loser. One easy extension to a ranking rule is to output the reverse order by which candidates are eliminated. Here, ties can be broken based on some tie breaking rule. The following example shows that this easy interpretation of STV as a ranking rule is not preserving single-peakedness.

\begin{example}\label{ex: STV not sp preserving}
    Consider the preference profile given by
    \begin{align*}
        1 \times ~~ &a \succ b \succ c \\
        2 \times ~~ &c \succ b \succ a. 
    \end{align*}
	STV first eliminates~$b$ with plurality score~$0$, then~$a$ (plurality score~$2$), which leaves~$c$ as a winner.
	The resulting order over the candidates is $c \succ a \succ b$.
	The preference profile is single-peaked only w.r.t.\ the axis~$a\rhd b\rhd c$ (and its reverse), but the output ranking is not.
\end{example}

Since this example can be arbitrarily extended by adding more candidates always in the same order to the bottom of each ranking, it generally holds that STV is not preserving single-peakedness on preference profiles with $m\geq3$ candidates.
}

\toappendix{
\subsection{Borda}\label{subsec: borda}
Borda's ranking rule is a positional scoring rule for which candidates receive scores according to their positions in the rankings of the voters. These scores then give a natural ranking over the candidates. More specifically, for a preference profile $\profile=(C,V)$ over $m$ candidates the Borda score of candidate $c$ is given by $B(c) = \sum_{v \in V} \sum_{c'\in C\setminus \{c\}} \mathds{I}(c\succ_v c')$.

The following example shows that the Borda ranking of a profile with only single-peaked input preferences is not necessarily single-peaked again.

\begin{example}\label{ex: Borda-not-sp-preserving-m=4}
    We consider the following preference profile:
    \begin{tabbing}
	    \hspace{1cm}\= $3\times$ \hspace{0.5cm}\= $a \succ b \succ c \succ d$ \\
        \> $2\times$ \> $c \succ d \succ b \succ a$.
    \end{tabbing}
	The Borda scores of the four candidates are: $B(a)=9$, $B(b)=8$, $B(c)=9$, $B(d)=4$, hence the Borda rankings are $\{c \succ a \succ b \succ d, a \succ c \succ b \succ d\}$.
	Note that this profile is single-peaked only w.r.t.\ the axis $a\rhd b\rhd c\rhd d$ (and its inverse), but the output rankings are not single-peaked w.r.t.\  this axis.
\end{example}

This example can be extended to any~$m>4$ by adding more candidates always in the same order to the bottom of each ranking, it generally holds that Borda's rule is not preserving single-peakedness on preference profiles with $m\geq4$ candidates.
However, for $m=3$ we can show that Borda's rule preserves single-peakedness. 
\begin{lemma}
	\label{lem:boada-sp-3}
	On preference profiles with at most~$3$ candidates, Borda's rule preserves single-peakadness.
\end{lemma}

\begin{proof}
	This result trivially holds for~$m\le2$ candidates.
	For~$m=3$ consider the following profile~$\profile$ that includes all possible single-peaked rankings for the candidates~$a$, $b$, and~$c$ w.r.t.\ axis~$a\rhd b\rhd c$.
	Here, $r_1, \dots, r_4 \geq 0$ signify the number of voters with the corresponding preferences.
	\begin{tabbing}
		\hspace{1cm}\= $r_1\times$ \hspace{0.5cm}\= $a \succ b \succ c$ \\
		\> $r_2\times$ \> $b \succ a \succ c$\\
		\> $r_3\times$ \> $b \succ c \succ a$\\
		\> $r_4\times$ \> $c \succ b \succ a$\\
	\end{tabbing}
	Then the Borda scores of the candidates are:
	\begin{itemize}
		\item[] $B(a) = 2r_1 + r_2$,
		\item[] $B(b) = r_1 + 2r_2 + 2r_3 + r_4$, and
		\item[] $B(c) = r_3 + 2r_4$.
	\end{itemize}
	Suppose $B(a)>B(b)$, then 
	\begin{align*}
		2r_1 + r_2 &> r_1 + 2r_2 + 2r_3 + r_4 \\
		r_1 &> r_2 + 2r_3 + r_4 .
	\end{align*}
	In consequence,
	\begin{align*}
		B(b)= r_1 + 2r_2 + 2r_3 + r_4 &> 3r_2 + 4r_3 + 2r_4 \\
		&> r_3 + 2r_4 = B(c),
	\end{align*} i.e., $B(a)>B(b)>B(c)$.

	Similarly, if $B(c)>B(b)$ (i.e., $r_4 > r_1 + 2r_2 + r_3$) then $B(c)>B(b)>B(a)$.
	Thus, the outcome Borda rankings in $\rrule(\profile)$ are either $a \succ b \succ c$ or $c \succ b \succ a$ or respect $b \succ c$ and $b \succ a$.
	Hence, Borda's rule outputs a single-peaked preference ranking for any preference profile $\profile$ with only single-peaked preference rankings over 3 candidates as input.
\end{proof}

	However, it turns out that even with~$3$ candidates Borda's rule is not extremist majority consistent.
	The following example illustrates our above observation that the Borda ranking is $a \succ b \succ c$ if and only if $r_1 > r_2 + 2r_3 + r_4$.
	That is, even if the majority opinion is one of the extreme opinions, this opinion is not necessarily a Borda ranking.

	\begin{example}\label{ex: Borda-not-extrem-majority-endorsing}
		We consider the following preference profile $\profile$:
		\begin{tabbing}
			\hspace{1cm}\= 4: \hspace{0.5cm}\= $a \succ b \succ c$ \\
			\> 1: \> $b \succ a \succ c$ \\
			\> 2: \> $b \succ c \succ a$
		\end{tabbing}
		The Borda scores of the candidates are: $B(a)=9$, $B(b)=10$, $B(c)=2$, such that the Borda ranking of $\profile$ is $b \succ a \succ c$ which differs from the majority opinion of $a \succ b \succ c$. 
	\end{example}
}

\section{Convergence of Single-Peaked Updates}\label{sec: convergence}
As we identified in the previous section, that only Kemeny's rule (or its equivalents on the single-peaked domain) and MMC preserve single peakedness, we consider the convergence of diffusion processes with these two rules as update rules on single-peaked opinions.  We show first for Kemeny updates and then for MMC updates (with tie-breaking as described in \cref{sec: prelim}) that any sequence of updates is finite, i.e., ends in a stable state after finitely many updates.

\subsection{Convergence of Kemeny Updates}

In the following, we outline how a stable state can be found in any network with single-peaked voter preferences that use Kemeny's rule as an update rule (while maintaining single-peakedness and using the above mentioned tie-breaking rule).

\begin{theorem}
\label{lem:kemeny-convergence}
Let~$G=(V, E)$ be a network with single-peaked preference profile~$\profile=(C, V)$.
Then any update sequence using Kemeny's rule is of length at most~$|E|\cdot\binom{|C|}{2}$.
That is, Kemeny's rule converges.
\end{theorem}
\begin{proof}
	We show that any update to~$G$ decreases the sum~$\sum_{\{u, v\} \in E} \kt(\succ_u, \succ_v)$ of the distances by at least~$1$.
	Consider an update by Kemeny's rule for a voter~$v \in V$ with opinion $\succ_v$ and let~$r$ be the newly obtained opinion of~$v$, i.e., the Kemeny ranking of the profile $\profile[N(v)]$ of $v$'s neighbours in $G$.
	Let~$\kappa := \sum_{u \in N(v)} \kt(\succ_u, \succ_v)$ and let~$\kappa^* := \sum_{u \in N(v)} \kt(\succ_u, r)$.
	As~$r$ is an opinion minimi\-sing Kendall's tau distance to the opinions of~$N(v)$, we have~$\kappa^* \le \kappa$.
	Further, only~$v$ changes its opinion, so the sum of Kendall's tau distances cannot increase.
	But if it does not change, then the opinion~$\succ_v$ is used as a tie breaker, and we have~$r =\,\, \succ_v$, that is, $v$ does not update its opinion.
	Lastly, no update sequence under Kemeny can have more than~$|E| \cdot \binom{|C|}{2}$ updates as for any two rankings~$r, r'$ we have~$\kt(r, r') \le \binom{|C|}{2}$.
\end{proof}

For a general complexity result, we show that we can efficiently compute an update based on Kemeny's rule.
For this, note that one of the ends of the single-peaked axis is always among the Condorcet losers.

\begin{lemma}
	\label{lem:kemeny-time}
	Computing an update step with Kemeny's rule takes~$O(|V|\cdot |C|^2)$ time.
\end{lemma}

\subsection{Convergence of MMC Updates}

Recall that $\top(v)$ is the peak in voter $v$'s single-peaked preference ranking. Further, denote by $\spdist(c_i,c_j) = |i-j|$ the distance between candidates $c_i$ and $c_j$ according to the single-peaked axis $c_1 \rhd \dots \rhd c_m$.
We show that if voters update their opinion by applying MMC, the total distance between the peaks of the opinions of neighbours can never increase.
\begin{lemma}[\appref{le: non-increasing-dist-sum-for-weakCWconsistant}]\label{le: non-increasing-dist-sum-for-weakCWconsistant}
Let~$G=(V, E)$ be a network with single-peaked preference profile~$\profile$. Under MMC updates the sum of distances $\sum_{\{v, v'\} \in E} \spdist(\top(v), \top(v'))$ cannot increase.
\end{lemma}

\appendixproof{le: non-increasing-dist-sum-for-weakCWconsistant}{
\begin{proof}
	When a voter $v$ updates their preference, then only $\spdist(\top(v), \top(v'))$ for voters $v' \in N(v)$ that are neighbours of $v$ can change.
	We claim that $\sum_{v' \in N(v)} \spdist(\top(v), \top(v'))$ can never increase.

	For a single-peaked preference profile $\profile$, MMC selects a (weak) Condorcet winner as the top ranked candidate. 
	By Lemma~\ref{le: extension-of-median-voter-thm} we know that the (weak) Condorcet winners of a single-peaked preference profile are those that are the peaks of median voters (when sorting the voters by their peaks according to the single-peaked axis) and any candidates that lie in between these peaks.
	We claim that any (weak) Condorcet winner $c_i$ minimises $\sum_{v' \in N(v)} \spdist(c_i, \top(v'))$. 

	The proof proceeds in the following two steps. First we establish that if there exist two median voters (i.e., the number of voters is even) with peaks $c_i$ and $c_j$ with $i<j$, then all candidates $c \in \{ c_i, c_{i+1}, \ldots, c_j\}$ have the same sum of distances to peaks of voters in $N(v)$. If there is an odd number of voters, then there exists exactly one median voter who's peak candidate is the (strict) Condorcet winner. Thus all Condorcet winners have the same sum of distances to peaks of voters in $N(v)$. Next we prove that this sum of distances minimises the sum of distances to peaks of voters in $N(v)$ among all candidates. Thus no matter which opinion voter $v$ comes from, the updated opinion cannot increase the sum of distances to the peaks of voters in $N(v)$.

	Let there exist two median voters with peaks $c_i$ and $c_j$ with $i<j$ and let $c \in \{ c_i, c_{i+1}, \ldots, c_j\}$. We show that
	\[ \sum_{v' \in N(v)} \spdist(c, \top(v')) = \sum_{v' \in N(v)} \spdist(c_i, \top(v')). \]
	We know that for all voters $v' \in N(v)$ their peak is either equal to $c_i$ or $c_j$ or is not a median peak, i.e., $\top(v') \in C\setminus \{ c_{i+1}, \ldots, c_{j-1}\}$. Let $\top_-$ be the set of peaks in $N(v)$ equal to $c_i$ or left of it on the single-peaked axis and let $\top_+$ be the set of peaks equal to $c_j$ or right of it on the single-peaked axis. By construction for all $v' \in N(v)$ either $\top(v') \in \top_-$ or $\top(v') \in \top_+$.
	Let~$N_-(v)=\{v' \in N(v) : \top(v')\in\top_-\}$ and~$N_+(v)=\{v' \in N(v) : \top(v')\in\top_-\}$ be the set of neighbours with peaks in~$\top^-$ and~$\top^+$, respectively.
	Then
	\allowdisplaybreaks
	\begin{align*}
		\sum_{v' \in N(v)} &\spdist(c_i, \top(v')) \\
		&= \sum_{v' \in N_-(v)} \spdist(c_i, \top(v')) + \sum_{v' \in N_+(v)} \spdist(c_i, \top(v')) \\
		&= \sum_{v' \in N_-(v)} \left(\spdist(c, \top(v')) + \spdist(c_i, c) \right)\\
		&\phantom{= \sum_{v' \in N_-(v)}} + \sum_{v' \in N_+(v)} \left( \spdist(c, \top(v')) - \spdist(c_i, c) \right) \\
		&= \frac{|N(v)|}{2} \cdot \spdist(c_i, c) + \sum_{v' \in N_-(v)} \spdist(c, \top(v')) \\
		&\phantom{= \sum_{v' \in N_-(v)}} - \frac{|N(v)|}{2} \cdot \spdist(c_i, c) \\
		&\phantom{= \sum_{v' \in N_-(v)}} + \sum_{v' \in N_+(v)} \spdist(c, \top(v')) \\
		&= \sum_{v' \in N_-(v)} \spdist(c, \top(v')) + \sum_{v' \in N_+(v)} \spdist(c, \top(v')) \\
		&= \sum_{v' \in N(v)} \spdist(c, \top(v')).
	\end{align*}

	We now show that any Condorcet winner minimises the distance to peaks of voters in $N(v)$.
	Let~$c_i$ be the peak of a median voter in~$N(v)$.
	Assume for contradiction that~$c_i$ does not minimise the distance $\sum_{v' \in N(v)} \spdist(c_i, \top(v'))$.
	Then there exists some candidate~$c_h$ such that
	\[\sum_{v' \in N(v)} \spdist(c_i, \top(v')) - \sum_{v' \in N(v)} \spdist(c_h, \top(v')) > 0.\]
	Note that in this case, if the peak of voter~$v$'s opinion changes from~$c_i$ to~$c_h$ in an update, the sum of distances would increase.
	Let~$w_\ell$ be the number of voters in~$N(v)$ that have~$c_\ell$ at the top of their preference ranking.
	Then, for a candidate~$c_k$ we have~$\sum_{v' \in N(v)} \spdist(c_k, \top(v')) = \sum_{\ell=1}^{m} |k-\ell| w_\ell$.
	Assume $i > h$.
	Then
	\begin{align*}
		0 &< \sum_{\ell=1}^m |i-\ell| w_\ell - \sum_{\ell=1}^m |h-\ell| w_\ell\\
		  &= \sum_{\ell=1}^{i-1} (i-\ell)w_\ell + \sum_{\ell=i}^{m} (i-\ell)w_\ell \\
		  &\phantom{= \sum_{\ell=1}^{i-1}} - \Big(\sum_{\ell=1}^h (h-\ell)w_\ell + \sum_{\ell=h+1}^{m} (\ell-h)w_\ell\Big)\\
		&= \sum_{\ell=1}^h (i-\ell-h+\ell)w_\ell + \sum_{\ell=h+1}^{i-1} (i-\ell-\ell+h)w_\ell \\
		&\phantom{= \sum_{\ell=1}^{i-1}} + \sum_{\ell=i}^m (\ell-i-\ell+h)w_\ell\\
		&= \sum_{\ell=1}^h (i-h)w_\ell - \sum_{\ell=i}^m (i-h)w_\ell + \sum_{\ell=h+1}^{i-1} (i+h-2\ell)w_\ell\\
		&= \sum_{\ell=1}^h (i-h)w_\ell - \sum_{\ell=i}^m (i-h)w_\ell + \!\sum_{\ell=h+1}^{i-1}\! (i-h)w_\ell \\
		&\phantom{= \sum_{\ell=1}^{i-1}} + \sum_{\ell=h+1}^{i-1} (i+h-2\ell-i+h)w_\ell\\
		&= (i-h)\big(\sum_{\ell=1}^{i-1} w_\ell - \sum_{\ell=i}^m w_\ell \big) + \sum_{\ell=h+1}^{i-1} (2h-2\ell)w_\ell.
		\stepcounter{equation}\tag{\theequation}\label{eq:weakCWconsistent-contradiction}
	\end{align*}
	As~$c_i$ is the peak of a median voter in~$N(v)$, we have~$\sum_{\ell = 1}^i w_\ell \ge \sum_{\ell=i+1}^{m} w_\ell$ and~$\sum_{\ell = 1}^{i-1} w_\ell \le \sum_{\ell = i}^m w_\ell$.
	Hence, the first term in \eqref{eq:weakCWconsistent-contradiction} is at most zero, and as~$w_\ell \ge 0$ for all~$1 \le \ell \le m$, the second term is at most zero as well.
	This is a contradiction to \eqref{eq:weakCWconsistent-contradiction} being greater than zero.
	The case~$i<h$ is analogous due to symmetry; thus a weak Condorcet winner $c_i$ always minimises the distance $\sum_{v' \in N(v)} \spdist(c_i, \top(v'))$, and our claim is proven.
\end{proof}
}

Note that while a single MMC update can be computed in polynomial time,~\cref{le: non-increasing-dist-sum-for-weakCWconsistant} does not give us a distance measure that constantly decreases.
Thus we cannot estimate how long an update sequence can be. Nevertheless, we now show the convergence of MMC as an update rule.

\begin{theorem}[\appref{lem:minimax-convergence}]
\label{lem:minimax-convergence}
Let~$G = (V, E)$ be a network with single-peaked preference profile $\profile$.
Then any update sequence using MMC is finite.
That is, MMC converges.
\end{theorem}

\begin{proof}[Proof sketch]
We prove the statement by induction over the number of candidates $m$. By Lemma~\ref{le: non-increasing-dist-sum-for-weakCWconsistant}, the sum of distances $\sum_{\{v, v'\} \in E} \spdist(\top(v), \top(v'))$ cannot increase with any update. In particular, in an update cycle the distances must stay equal for every update. The top choice of a voter after that voter updated must always be a weak Condorcet winner and, because of the tie breaking rule, must remain the same Condorcet winner throughout all updates.
Furthermore, because of the tie breaking rule, and because the only single-peaked rankings with $c_1$, respectively $c_m$, on the top are the extreme opinions, Lemma~\ref{le: non-increasing-dist-sum-for-weakCWconsistant} implies that there can be no update cycle for which opinions are switched from or to an extreme opinion.  
The base case of the induction, $m=2$, is thus settled.
For the induction step assume there exist no cycles of voter updates in any network with single-peaked opinions over $m$ candidates, but that there exists an update cycle $K$ in~$G$ with single-peaked preference profile $\profile$ over $m+1$ candidates. We show that by eliminating the last candidate of the single-peaked axis of $\profile$, the relative ranks of the candidates remain the same for every voter update in $K$ and thus $K$ must be an update cycle on $G$ with a single-peaked preference profile over $m$ candidates -- a contradiction.
\end{proof}

\appendixproof{lem:minimax-convergence}{
\begin{proof}
We prove the statement by induction over the number of candidates $m$. By Lemma~\ref{le: non-increasing-dist-sum-for-weakCWconsistant}, the sum of distances $\sum_{\{v, v'\} \in E} \spdist(\top(v), \top(v'))$ cannot increase with any update. In particular, in an update cycle this means that the distances must stay equal for every update. The top choice of a voter after that voter updated must always be a weak Condorcet winner and, because of the tie breaking rule, must remain the same Condorcet winner throughout all updates.

Furthermore, because of the tie breaking rule, and because the only single-peaked rankings with $c_1$, respectively $c_m$, on the top are the extreme opinions, Lemma~\ref{le: non-increasing-dist-sum-for-weakCWconsistant} implies that there can be no update cycle for which opinions are switched from or to an extreme opinion.  
 The base case of the induction, $m=2$, is thus settled.
 
 Assume now there exist no cycles in any network with single-peaked opinions over $m$ candidates.
 Suppose there exists a network $G= (V,E)$  with single-peaked preference profile $\profile$ over $m+1$ candidates that admits a cycle of updates of voters $(v_1, \dots, v_k)$ with $v_1 = v_k$  such that after these updates the preference profile of $G$ is the same as before, $\profile$. Consider the profile $\profile'$ over $m$ candidates that results from $\profile$ by deleting candidate $c_{m+1}$ from every voters preference ranking. 
 We show in the following that $(v_1, \dots, v_k)$ is a cycle of voter updates in $G$ which, starting from profile $\profile'$ leads to the same profile $\profile'$ again. This is a contradiction to the induction assumption.
 
 Let us denote by $\profile_i$ the profile resulting from the updates of voters $v_1, \dots, v_i$ in $G$ starting from $\profile$. Similarly, let $\profile_i'$ be the profile after the first $i$ updates starting from $\profile'$ in $G$. Let $l \in [k]$. Since $c_{m+1}$ is never at the top of any of the voters $v_i$'s opinions in any $\profile_l$, $c_m$ can never be defeated by $c_{m+1}$ in a pairwise comparison in any $\profile_l[N(v_i)]$.
 Also, by Lemma~\ref{lem:mmc_worst_defeat}, none of the candidates $c_1, \dots, c_{m-1}$ can experience their worst defeat only against $c_{m+1}$ in any of the profiles $\profile_l[N(v_i)]$. 
 Thus, for every voter $v_i$ every candidate $c_1, \dots, c_{m-1}$ admits the same MMC score in any $\profile_l[N(v_i)]$ as in  $\profile'_l[N(v_i)]$.
The MMC score of $c_m$ might decrease in case its worst 'defeat' was only against $c_{m+1}$, and stays the same otherwise. Consider the update of one voter $v_i$ at the $i$-th step of the update cycle and assume $c_m$'s worst 'defeat' was only against $c_{m+1}$ in  $\profile_{i-1}[N(v_i)]$. Because the only single-peaked opinion in which $c_{m+1}$ is ranked over $c_m$ is $r^{\downarrow}$, every other candidate has at least as bad of a defeat against $c_{m+1}$ and their MMC score is at least as high as $c_m$'s.
Let the set of voters in $N(v_i)$ with opinion $r^{\downarrow}$ be $M$. Suppose there exists another candidate $c_j$ that has exactly the same MMC score in  $\profile_{i-1}[N(v_i)]$ as $c_m$ and is top ranked in $v_i$'s opinion in  $\profile_{i-1}$.
Because $c_{m+1}$ is not top ranked in $v_i$'s opinion in  $\profile_{i-1}$, $c_{m+1}$ must have a higher MMC score than $c_j$ and $c_m$. Thus, there must exist enough voters $M'$ in $\profile_{i-1}[N(v_i)]$ that rank $c_{m+1}$ lower than some other candidate, such that $c_{m+1}$'s MMC score is worse than $c_m$'s, i.e., with $|M'|-|N(v_i)\setminus M'|>|M|-|N(v_i)\setminus M|$. In particular, voters in $M'$ must rank $c_j \succ c_m \succ c_{m+1}$ so that $c_j$'s worst defeat is as high as $c_m$'s. This is a contradiction to $c_m$'s worst defeat being only against $c_{m+1}$ because it is beaten worse by $c_j$.
Thus we can assume that $c_m$ is top ranked in $v_i$'s opinion in all profiles  $\profile_l[N(v_i)]$. In $\profile'_{i-1}[N(v_i)]$, $c_m$'s MMC score decreased compared to $\profile_{i-1}[N(v_i)]$, but the MMC score of all other candidates stayed the same. Thus in $\profile_{i}'[N(v_i)]$, $c_m$ is still top ranked in voter $v_i$'s preference as well as in every other profile $\profile'_l[N(v_i)]$. 

Hence, in case $(v_1, \dots, v_k)$ is a cycle of voter updates in $G$ which, starting from profile $\profile$ leads to the same profile $\profile$ again, the same must be true for profile $\profile'$.
\end{proof}
}

\section{Maximally Spreading Extreme Opinions}\label{sec: spreading extreme opinions}

Especially in the political context, studying extreme opinions and their spread on social networks is often seen as more important than the spread of more centralist views. Platforms like Facebook, Youtube or Instagram invest vast resources to combat the spread of hateful and extreme opinions.

Recall, that for single-peaked axis~$c_1 \rhd \dots \rhd c_m$
the rankings~$r^\uparrow = c_1 \succ \dots \succ c_m$ and~$r^\downarrow = c_m \succ \dots \succ c_1$ are called the \emph{extreme opinions}, as they rank the extreme candidates on the single-peaked axis highest.
In the following, we consider under which conditions an extreme opinion can be maximally spread in a single-peaked preference network. For this purpose we define the notion of \emph{extremist majority consistency} for ranking rules that preserve single-peakedness -- a property that holds for both Kemeny's rule and MMC. We then formulate a general algorithm for extremist majority consistent rules that maximally spreads extreme opinions.

\subsection{Extremist Majority Consistency}\label{sec: general results extr maj conist}
Informally, extremist majority consistency states that an extreme opinion can only be adopted by an agent when a majority of the agent's neighbours hold this opinion.
\begin{definition}
	Let $\profile=(C,V) \in \mathcal{L}_{sp}^\rhd(C)$.
	We call a ranking rule $\rrule$ that preserves single-peakedness \emph{extremist majority consistent} if for any extreme opinion~$r^* \in \{r^\uparrow, r^\downarrow\}$ we have
	\begin{align*}
	    r^* \in \rrule(\profile) \iff & \text{ a (weak) majority in $V$ has opinion $r^*$}.
	\end{align*} 
\end{definition}

One might think that for extremist majority consistency it is sufficient that a ranking $\rrule$ preserves single-peakedness and is weak Condorcet winner consistent.
However, this is not the case since these two properties alone do not enforce that there exists an output ranking that is single-peaked \emph{and} ranks the desired weak Condorcet winner ($c_1$ or $c_m$) highest.

We proceed by showing that Kemeny's update rule and MMC are extremist majority consistent.

\begin{theorem}\label{le: update-to-extremes-kemeny}
Kemeny's rule is extremist majority consistent.
\end{theorem}
\begin{proof}
	Due to symmetry we need to prove the statement only for the extreme opinion~$r^\uparrow$.
    Assume that at least half of the voters in $\profile=(C,V)$ have opinion~$r^\uparrow$. By definition of Kemeny's rule, $r^\uparrow$ is among the output of the rule.
	For the other direction assume that $r^\uparrow$ is a Kemeny ranking and consider the ranking $r' = c_2 \succ c_1 \succ c_3 \succ \dots \succ c_m$, which is the only single-peaked ranking with $\kt(r^\uparrow, r') = 1$.
	Then $\kt(r^\uparrow, r) \geq \kt(r', r) +1$ for all single-peaked rankings $r \in \mathcal{L}_{sp}^\rhd(C) \setminus \{r^\uparrow\}$. Also:
\begin{align*}  
	&
	\kt(r', \profile)  
            = \sum\nolimits_{v \in V} \kt(r', \succ_v) \\ 
	    &= |\{v \in V \mid \succ_v=r^{\uparrow}\}| + \sum\nolimits_{v \in V : \succ_v\ne r^\uparrow} \kt(r', \succ_v)  \\ 
            &\leq |\{v \in V \mid \succ_v=r^{\uparrow}\}| + \sum\nolimits_{v \in V : \succ_v\ne r^\uparrow} \kt(r^{\uparrow}, \succ_v) - 1  \\ 
            &\leq |\{v \in V \mid \succ_v=r^{\uparrow}\}| + \kt(r^{\uparrow}, \profile) - |\{v \in V \mid \succ_v \neq r^{\uparrow}\}|  \\ 
            &\leq |\{v \in V \mid \succ_v=r^{\uparrow}\}| + \kt(r', \profile)   - |\{v \in V \mid \succ_v \neq r^{\uparrow}\}|
\end{align*}
and hence, $|\{v \in V \mid \succ_v=r^{\uparrow}\}| \geq |\{v\in V \mid \succ_v \neq r^{\uparrow}\}|$.
\end{proof}

Note that for non-extreme opinions this equivalence need not hold, i.e., a Kemeny ranking need not be supported by a majority of the voters. However, if the majority of voters has opinion $r$, then $r$ is a Kemeny ranking (see \cref{apsec: Kemeny-non-extreme-opinions}).

\toappendix
{
\subsection{Under Kemeny Updates Non-Extreme Opinions are not Majority Consistent}\label{apsec: Kemeny-non-extreme-opinions}
The following example shows that there exist Kemeny rankings that are not supported by a majority of voters.
\begin{example}
Consider the following single-peaked preference profile.
    \begin{align*}
        2 \times ~~ &a \succ b \succ c \\
        1 \times ~~ &b \succ c \succ a \\
        1 \times ~~ &c \succ b \succ a 
    \end{align*}
    
One of the Kemeny rankings in this case is $b \succ a \succ c$ even though none of the voters' preference rankings coincide with $b \succ a \succ c$.
\end{example}

Next, we show the reverse. That is, any ranking that is supported by a majority is a Kemeny ranking. 
\begin{lemma}
\label{le: update-to-majority-opinion-Kemeny}
    Let $\profile$ be a single-peaked preference profile. If the majority of voters in $\profile$ has opinion $r$, then $r$  is a Kemeny ranking of $\profile$.
\end{lemma}
\begin{proof}
Suppose that the majority of voters $V$ in $\profile$ has opinion $r$ and that $r$ is not a Kemeny ranking of $\profile$. Let $r'$ be some Kemeny ranking of $\profile$, let $d=\kt(r,r')$ and let $M$ be the set of voters in $V$ with opinion $r$. The Kemeny score of $r'$ must be smaller than the one of $r$ with respect to $\profile$. 
Furthermore, for every voter $v \in V$ we have $\kt(r,\succ_v)\leq \kt(r',\succ_v)+d$. 
This yields a contradiction:
\begin{align*}
\sum_{v \in V} \kt(r,\succ_v) &= \sum_{v \in V\setminus M} \kt(r,\succ_v) \\
&\leq \sum_{v \in V\setminus M} (\kt(r',\succ_v) + d) \\
&= \sum_{v \in V\setminus M} \kt(r',\succ_v) + |V\setminus M|\cdot d\\
&\le \sum_{v \in V\setminus M} \kt(r',\succ_v) + |M|\cdot d\\
&= \sum_{v \in V\setminus M} \kt(r',\succ_v) + \sum_{v \in M} \kt(r',\succ_v)\\
&= \sum_{v \in V} \kt(r',\succ_v)
\end{align*}  
    \end{proof}
 }

MMC, just like Kemeny's rule, outputs an extreme opinion if and only if a (weak) majority of voters has this opinion.
\begin{theorem}[\appref{prop: mmc-extremist-majority-consistent}]\label{prop: mmc-extremist-majority-consistent}
MMC is extremist majority consistent.
\end{theorem}

\appendixproof{prop: mmc-extremist-majority-consistent}{
\begin{proof}
We prove that~$r^\uparrow = c_1 \succ \dots \succ c_m \in R_\mmc$ if and only if a weak majority of voters in $\profile$ has opinion $r^\uparrow$.
The equivalence then also holds for $r^\downarrow$ due to symmetry.

First, let $r^\uparrow$ be among the winners of $\profile$ according to MMC, that is, for all~$j \in \{2, \dots, m\}$ we have
\begin{equation}\label{eq:update-to-extremes-MMC-1}
    \mmc(c_j) \geq \mmc(c_1).
\end{equation}
Towards a contradiction assume that a strict majority of voters has a preference that differs from $r^\uparrow$.
Since $r^\uparrow$ is the only single-peaked ranking that ranks $c_1$ over $c_2$, this implies 
$|\{ v \in V: c_2 \succ_{v} c_1 \}| > |\{ v \in V : c_1 \succ_{i} c_2 \}|$ and equivalently $\pop(c_2, c_1) > 0$.
We claim that for $i\in[m-1]$ we have $\pop(c_{i+1}, c_i)>0$.
By \Cref{lem:mmc_worst_defeat} this is true for $i=1$.
So suppose $\pop(c_{i+1}, c_i)>0$ for some $i<m-1$.
By \cref{lem:mmc_worst_defeat} we have $\mmc(c_{i+1}) = \max \{\pop(c_i, c_{i+1}), \pop(c_{i+2}, c_{i+1})\}$.
As $\pop(c_i, c_{i+1}) = -\pop(c_{i+1}, c_i) < 0$ and $\mmc(c_{i+1}) \ge \mmc(c_1) > 0$ by \cref{eq:update-to-extremes-MMC-1}, we have $\pop(c_{i+2}, c_{i+1})>0$.
But then $\mmc(c_m) = \pop(c_{m-1}, c_m) = -\pop(c_m, c_{m-1}) < 0$, which contradicts \cref{eq:update-to-extremes-MMC-1} and completes the first direction of the proof.

Now let a (weak) majority of voters have preference $r^\uparrow$.
Then for each $j \in [m-1]$ we have 
\[ |\{ v \in V : c_j \succ_v c_{j+1} \}| \geq |\{ v \in V : c_{j+1} \succ_v c_{j} \}| \]
and thus $\pop(c_j, c_{j+1}) \geq 0$.
By \Cref{lem:mmc_worst_defeat} a worst defeat a candidate $c_j$ experiences is against either $c_{j-1}$ or $c_{j+1}$. Since $\pop(c_j, c_{j+1}) \geq 0$, either candidate $c_j$ wins or ties against every other candidate or its worst defeat is against $c_{j-1}$ (and this is not a tie, i.e., $\pop(c_{j-1}, c_{j}) > 0$).
Towards a contradiction now assume that $r^\uparrow$ is not an MMC winner ranking, i.e., there is a candidate $c_j$ and some $k \in \mathbb{N}_{>0}$ such that
\[ \max_{b \in C} \pop(b, c_j) > \max_{b \in C} \pop(b, c_{j+k}). \]
Due to the strict inequality, $c_j$ cannot win or tie against everyone, so we obtain $\pop(c_{j-1}, c_j) > 0$.
Further, by \Cref{lem:mmc_worst_defeat} the above inequality is equivalent to
\begin{align*}
	&\max \left\{ \pop(c_{j-1}, c_{j}), \pop(c_{j+1}, c_{j}) \right\}\\ >& \max \left\{ \pop(c_{j+k-1}, c_{j+k}), \pop(c_{j+k+1}, c_{j+k}) \right\}.
\end{align*}
Since $\pop(c_j, c_{j+1}) \geq 0$ (and thus $\pop(c_{j+1}, c_{j}) \leq 0$) the maximum on the left is obtained by $\pop(c_{j-1}, c_{j})$. In particular $\pop(c_{j-1}, c_{j}) > \pop(c_{j+k-1}, c_{j+k})$. This means that more voters rank $c_{j-1} \succ c_j$ than $c_{j+k-1} \succ c_{j+k}$.
But all voters with preference $r^\uparrow$ rank $c_{j-1} \succ c_j$ and $c_{j+k-1} \succ c_{j+k}$. Thus at least one of the remaining voters has to rank $c_{j-1} \succ c_j$ and $c_{j+k} \succ c_{j+k-1}$.
But these preferences are a violation of the single-peaked axis and thus can not occur in $\profile$. This violates the assumption that $r^\uparrow$ is not an MMC winner ranking.
   \end{proof}
 }

\subsection{Maximally Spreading Extreme Opinions for Extremist Majority Consistent Update Rules}

We will next prove that for an extremist majority consistent ranking rule $\rrule$ and an extreme opinion~$r^*$ the following greedy update sequence~$\sigma^*$ reaches a stable state that maximises the number of voters with opinion~$r^*$.
Note that due to step (3) below this sequence is only well defined if there exists a finite sequence of updates w.r.t.\ $\rrule$ that leads to a stable state for any network.

	\paragraph{\textsc{Greedy Sequence} $\sigma^*$ \textsc{for extreme opinion} $r^*$}

	\begin{enumerate}[(1)]
		\item Update every non-stable voter with opinion~$r' \ne r^*$ to opinion~$r^*$ if possible.
		\item Update every non-stable voter with opinion~$r^*$.
		\item Stabilize network: update non-stable voters with opinions~$r' \ne r^*$.
	\end{enumerate}

A similar algorithm was used before in the setting with two competing opinions in \cite{auletta2015,BE17} and in the setting of \emph{discrete preference games} in \cite{chierichetti2013discrete}.
In particular, \citet[Proposition 1]{BE17} show that in the binary model with only two possible opinions, first updating every non-stable voter with the second opinion to the first opinion, and then updating every non-stable voter with the first opinion to the second opinion yields a stable outcome with a maximum number of voters having the first opinion.
Similarly, we derive the following result.
\begin{lemma}[\appref{lem:extremist-bin}]
	\label{lem:extremist-bin}
	Let $G=(V,E)$ be a preference network with single-peaked preference profile~$\profile$ and let~$\rrule$ be an extremist majority consistent ranking rule.
	The subsequence $\sigma$ of updates in~$\sigma^*$ that only runs steps (1) and (2) satisfies:
	\begin{enumerate}[(i)]
		\item\label{item:extremist-length} The length of~$\sigma$ is at most~$2|V|$.
		\item\label{item:extremist-changes} Each voter changes their opinion at most twice on $\sigma$.
		\item\label{item:extremist-time} Computing~$\sigma$ requires~$O(|V|^2)$ executions of~$\rrule$.
		\item\label{item:extremist-stable} After~$\sigma$, the set~$V^*$ of voters that still have opinion~$r^*$ is stable.
		\item\label{item:extremist-equal} For every sequence~$\sigma'$ that maximises the number of voters with opinion~$r^*$ such that every such voter is stable, the set of voters with opinion~$r^*$ is equal to~$V^*$.
	\end{enumerate}
\end{lemma}
\appendixproof{lem:extremist-bin}{
\begin{proof}
	Properties~\eqref{item:extremist-length} and~\eqref{item:extremist-changes} are immediate.
	As finding all non-stable voters after each of the~$O(|V|)$ updates takes~$O(|V|)$ executions of~$\rrule$ each, Property~\eqref{item:extremist-time} holds.
	After completing~$\sigma$, a non-stable voter with opinion~$r^*$ contradicts that step (2) was completed, thus Property~\eqref{item:extremist-stable} holds.
	Let us focus on Property~\eqref{item:extremist-equal}.
    Denote by~$\overline{r^*}$ the set of single-peaked rankings that are different from~$r^*$.
    Note that after completing~$\sigma$, apart from every voter with opinion~$r^*$ being stable, there is also no voter~$v$ with an opinion in~$\overline{r^*}$ that can be updated to~$r^*$.
	This is true since the number of $v$'s neighbours with opinion~$r^*$ cannot have increased after step~(1), and~$v$ was not updated to~$r^*$ within step~(1).
	Now let~$\sigma'$ be as in Property~\eqref{item:extremist-equal}.
	We show that (a)~each voter that was updated from an opinion in~$\overline{r^*}$ to~$r^*$ by~$\sigma'$ also was updated from an opinion in~$\overline{r^*}$ to~$r^*$ by~$\sigma$ and that (b)~each voter that was updated from~$r^*$ to an opinion in~$\overline{r^*}$ by~$\sigma$ has an opinion in~$\overline{r^*}$ after completing~$\sigma'$.
	For~(a), assume for contradiction that there is a voter that updated from~$\overline{r^*}$ to~$r^*$ by~$\sigma'$, but not under~$\sigma$.
	Let~$v'$ be the first such voter in~$\sigma'$, that is, the opinion of~$v'$ is changed in step~$k$ of~$\sigma'$ and all voters that were updated from~$\overline{r^*}$ to~$r^*$ in some step~$k' < k$ of~$\sigma'$ also were updated from~$\overline{r^*}$ to~$r^*$ by~$\sigma$.
	Then all neighbours of~$v'$ that have opinion~$r^*$ after the~$(k-1)$-st update of~$\sigma'$ (which must be a majority because of extremist majority consistency of $\rrule$) also have opinion~$r^*$ at the end of step (1) of~$\sigma$.
	Hence, a majority of the neighbours of~$v'$ would have opinion~$r^*$ at the end of step (1), but~$v'$ is not updated by~$\sigma$ -- a contradiction to $\rrule$ being extremist majority consistent.
	The proof for~(b) is analogous.
\end{proof}
} %

If the ranking rule in use converges, then we are guaranteed that step~(3) has a finite number of updates.
What remains to show is that, in this case, running step~(3) does not affect the stability of the voters with opinion~$r^*$ and that the set $V^*$ of voters with opinion~$r^*$ after completing~$\sigma^*$ is the same as for any other sequence that reaches a stable state and maximises the spread of opinion~$r^*$.

\begin{proposition}[\appref{prop:extremist-greedy}]
	\label{prop:extremist-greedy}
	Let~$G=(V,E)$ be a network with single-peaked preference profile~$\profile$ and an extremist majority consistent and converging ranking rule~$\rrule$, and let~$r^*$ be an extreme opinion.
	Let~$V^*$ be the set of voters with opinion~$r^*$ after completing steps (1) and (2) of~$\sigma^*$.
	The sequence~$\sigma^*$ satisfies:
	\begin{enumerate}[(i)]
		\item After completing~$\sigma^*$, every voter is stable.
		\item During step (3) of~$\sigma^*$, $V^*$ remains unchanged.
		\item For every sequence~$\sigma'$ that maximises the number of voters with opinion~$r^*$ such that every voter (independent of its opinion) becomes stable, the set of voters with opinion~$r^*$ is equal to~$V^*$.
	\end{enumerate}
\end{proposition}
\appendixproof{prop:extremist-greedy}{
\begin{proof}
	As every voter in~$V^*$ is stable, and every voter in~$V\setminus V^*$ is made stable during step~(3), Property (i) holds.
	Property~(ii) follows from $\rrule$ being extremist majority consistent. 
	Property~(iii) is an immediate consequence of \cref{lem:extremist-bin}, Property (ii) and the convergence of~$\rrule$.
\end{proof}
} %

\begin{corollary}
	\label{thm:kemeny-extremist-efficient}
	Let~$G = (V, E)$ be a network with preference profile~$\profile=(C,V)\in \mathcal{L}_{sp}^\rhd(C)$.
	Then one can compute a sequence of Kemeny updates that maximises the number of voters with opinion~$r^* \in \{r^\uparrow,r^\downarrow\}$ in~$O(|V|^3\cdot |C|^4)$ time.
\end{corollary}

\begin{corollary}
	\label{thm:MMC-extremist-sequence}
	Let~$G = (V, E)$ be a network with preference profile~$\profile=(C,V)\in \mathcal{L}_{sp}^\rhd(C)$.
		Then there exists a finite sequence of MMC updates that maximises the number of voters with opinion~$r^* \in \{r^\uparrow,r^\downarrow\}$ and ends in a stable state.
\end{corollary}

\section{Conclusion}\label{sec: conclusion}
Motivated by the question of how single-peaked opinions propagate in social networks we studied opinion diffusion processes in this domain.
We investigated which rules admit a single-peaked outcome given single-peaked preferences.
For these rules we then established convergence for arbitrary update sequences and provided an algorithm that outputs an update sequence to optimally spread an extreme opinion in any preference network.

We conclude by suggesting future research directions.
While single-peakedness is a well-established domain restriction, the study of opinion diffusion under single-crossing or single-caved domains is also well motivated.
Further, our initial study on opinion diffusion may be extended: can we efficiently compute update sequences to optimally spread non-extreme opinions?
Lastly, while many ranking rules coincide with Kemeny's rule in the single-peaked domain, this is not necessarily true for preference profiles with bounded single-peaked width \cite{CGS12}.
In terms of opinion diffusion, while a Kemeny ranking can be computed efficiently whenever the single-peaked width is small~\cite{CGS13}, it is not clear whether, e.g., our greedy sequence is applicable in this scenario.

\section*{Acknowledgements}
Anne-Marie George was supported under the Norwegian Research Council Grant No. 302203
“Algorithms and Models for Socially Beneficial Artificial Intelligence”.
Jonas Israel was supported by the Deutsche Forschungsgemeinschaft under grant \mbox{BR~{4744/2-1}}.

\bibliography{references}

\clearpage
\appendix
\section*{Appendix}
\appendixProofText

\end{document}